\documentclass[11pt]{article}

\usepackage[utf8]{inputenc}
\usepackage{color}
\usepackage{amsmath}
\usepackage{amssymb}
\usepackage{amsfonts}
\usepackage{hyperref}
\usepackage[round]{natbib}
\usepackage{fullpage}
\usepackage{amsthm}

\usepackage{graphicx}
\usepackage{tikz,tkz-tab}
\usepackage[metapost]{mfpic}
\opengraphsfile{figs} 
\usepackage{pgfplots, pgfplotstable}
\usetikzlibrary{plotmarks,positioning,shapes,arrows,backgrounds,patterns}
\usepackage{algorithmic}
\usepackage{algorithm}

\DeclareMathAlphabet{\mathpzc}{OT1}{pzc}{m}{it}

\def\alphag{\boldsymbol{\alpha}}

\def\B{{\rm B}}

\def\Xb{{\bf X}}

\def\Jb{{\bf J}}

\def\E{{\rm E}}

\def\Var{{\rm var}}

\def\R{\mathbb R}

\def\1{\mathbf{1}}

\hypersetup{
    unicode=false,          
    pdftoolbar=true,        
    pdfmenubar=true,        
    pdffitwindow=false,     
    pdfstartview={FitH},    
    colorlinks=true,       
    linkcolor=black,          
    citecolor=blue,        
    filecolor=magenta,      
    urlcolor=green           
}

\newtheorem{lemma}{Lemma}[section]
\newtheorem{definition}{Definition}[section]

\newtheorem{proposition}{Proposition}[section]

\numberwithin{equation}{section}

\begin{document}


\title{Quasi-Systematic Sampling From a Continuous Population}
\author{Matthieu {\sc Wilhelm}, Yves {\sc Tillé} and Lionel {\sc Qualité}\\
{}\\
\normalsize Institute of Statistics\\
\normalsize Faculty of Sciences\\
\normalsize University of Neuch\^atel\\
{}\\
{}\\
}
\maketitle
\begin{abstract}
A specific family of point processes are introduced that allow to select samples for the purpose of estimating the mean or the integral of a function of a real variable. These processes, called quasi-systematic processes, depend on a tuning parameter $r>0$ that permits to control the likeliness of jointly selecting neighbor units in a same sample. When $r$ is large, units that are close tend to not be selected together and samples are well spread. When $r$ tends to infinity, the sampling design is close to systematic sampling. For all $r > 0$, the first and second-order unit inclusion densities are positive, allowing for unbiased estimators of variance. 

Algorithms to generate these sampling processes for any positive real value of $r$ are presented. When $r$ is large, the estimator of variance is unstable. It follows that $r$ must be chosen by the practitioner as a trade-off between an accurate estimation of the target parameter and an accurate estimation of the variance of the parameter estimator. The method's advantages are illustrated with a set of simulations.
\end{abstract}




\section{Introduction}
We propose to use a specific family of point processes to select samples for the purpose of estimating the mean or the integral of a function of a real variable. We draw a parallel with sampling designs which are themselves point processes on finite spaces. Systematic sampling is widely used in finite population. It  has been introduced by \cite{mad:mad:44} and \cite{mad:49}. It is easily implemented and, by spreading the sample over the population, it results in precise mean and total estimators when the variable of interest is similar for neighboring units. The main drawback of systematic sampling is that most of the unit joint inclusion probabilities are null, making it impossible to estimate the variance of the Horvitz-Thompson estimator without bias \citep[see][]{hor:tho:52}. 

The aim of this paper is to develop a method that is a compromise between a base point process such as the Poisson process or the binomial process and the systematic process for sample selections in a continuous population. A similar objective is pursued in \cite{Breidt95} in a finite population setting supported by a superpopulation model. \cite{Breidt95} considers one-per-stratum sampling designs from a population that is split into strata of $a$ successive units where $a$ divides the population size. He introduces a class of sampling procedures that encompasses systematic sampling with constant rate $1/a$ and simple random sampling of one unit per stratum.

Point processes, that we refer to as \emph{sampling processes} in the context of sampling, are the subject of a vast literature \citep[see for example][and references therein]{Daley02, Daley08}. \cite{Cordy93} and \cite{dev:89x} introduced independently the continuous analogue to the Horvitz-Thompson estimator for infinite population sampling. Different communities have studied point processes: mathematical physicists, probabilists and statisticians. A detailed state of the art in the study and simulation of some complex point processes can be found in \cite{Moller03, Moller07}. Many simulation methods for point processes are implemented in the $\mathbf{R}$ package \texttt{spatstat} \citep{spatstat}.

We introduce a new family of sampling methods that enable to continuously tune the distance between units in the sample. These processes allow to obtain small probabilities of jointly selecting neighboring units. These sampling methods are particularly efficient when the function of interest is smooth. Moreover, joint inclusion densities are positive and it is possible to estimate the sampling variance without bias.

The paper is organized as follows: in Section~\ref{S3}, we give a definition of sampling processes in continuous populations and we define the Poisson process, the binomial process and the systematic process. Important results of renewal process theory are recalled in Section~\ref{S4}. In Section~\ref{S5}, we define the systematic-Poisson and the systematic-binomial processes with tuning parameter $r$, and compute the joint densities. Section~\ref{S6} contains proofs for the asymptotic processes when $r$ tends to infinity. Simulations are presented in Section~\ref{S7} and our ideas on the choice of the tuning parameter in Section~\ref{St}. Finally, we give a brief discussion of the method and its advantages in Section~\ref{S8}.

\section{Sampling from a continuous population}\label{S3}

Following \cite{Macchi75} \citep[see also][]{moy:62}, a finite sample of size $n$ from a bounded and open subset $\Omega$ of $\R$ is a collection of units $X= \{x_{1},\dots,x_{n}\}$ without consideration for the order of the $x_{i}$'s. This definition matches those commonly used in finite population sampling \citep[see for example][for an introduction to finite population sampling theory]{coc:77}. A sampling process is a probability distribution on the space $\mathcal{S}$ of all such collections, for all $n\in \mathbb{N}$. Note that it is not directly a distribution on $\Omega^{\mathbb{N}}$ equipped with the tensor product of Borel sigma algebras $\mathcal{B}(\Omega)$ as the sample units are not ordered. An extensive discussion on the definition of a sampling point process on $\Omega$ and the corresponding symmetric measure on $(\Omega^{\mathbb{N}},\mathcal{B}^{\otimes \mathbb{N}}(\Omega))$ is given in \cite{Macchi75}. It is sufficient for our purpose to know that a sampling point process is a probability distribution on $(\mathcal{S},\mathcal{B})$ where $\mathcal{S}=\bigcup_{n\in \mathbb{N}}\Omega^{n}/\mathcal{R}^{n}$, with $x$ and $y$ in $\Omega^{n}$ being in the same class for the equivalence relation $\mathcal{R}^{n}$ if $x$ is a permutation of elements of $y$, and $\mathcal{B}$ is the sigma algebra generated by the family of counting events: 
$$
\left\{s \in \mathcal{S} \mbox{ such that }N(s,A)=p,\ A\in \mathcal{B}(\Omega),\ p\in \mathbb{N}\right\},
$$
and $N(s,A)$ is the number of elements of $s$ that are in $A$.

The first and second factorial moment measures of a sampling point process $X$ \citep[][]{moy:62} are defined respectively as 
$$
M_{1}=\left(
\begin{array}{ccc}
\mathcal{B}(\Omega) & \rightarrow & \R_{+} \\
      A & \mapsto & \E\left[N(X,A)\right]
\end{array}
\right),
$$
where $N(X,A)$ is the random number of elements of $X$ that are in $A$, and the second factorial moment measure is the extension to $\mathcal{B}(\Omega)^{\otimes 2}$ of
$$
M_{2}=\left(
\begin{array}{ccc}
\mathcal{B}(\Omega) \times \mathcal{B}(\Omega) & \rightarrow & \R_{+} \\
      A\times B & \mapsto & \E\left[N_{2}(X,A\times B)\right]
\end{array}
\right),
$$
where $N_{2}(X,A\times B)$ is the random number of pairs $(x_{i},x_{j})$, $i\neq j$ of elements  of $X$ such that $x_{i}\in A$ and $x_{j}\in B$. 

We call first and joint (second) order inclusion densities the respective densities of $M_{1}$ and $M_{2}$ with respect to the Lebesgue measure on $\Omega$ and $\Omega^{2}$ when they exist. In that case, the first order inclusion density $\pi$ is such that $M_{1}(A)=\int_A \pi(x) dx$, for all $A \in \mathcal{B}(\Omega)$, and the second-order inclusion density $\pi^{(2)}$ satisfies $M_{2}(A\times B)=\int_A\int_B \pi^{(2)}(x,y) dx dy$ for all  $A\times B \in \mathcal{B}(\Omega) \times \mathcal{B}(\Omega)$. Heuristically, the term $\pi(x)dx$ can be viewed as the probability that one unit of the sample lies between $x$ and $x+dx$, and $\pi^{(2)}(x,y) dx dy$ as the probability that one unit of the sample lies between $x$ and $x+dx$ and another between $y$ and $y+dy$, disregarding what happens outside of these sets. Likewise, one can define $k-$th order factorial moments and, when they exist, inclusion densities for $k\geq 3$.

We now turn to the problem of estimating the mean of a Lebesgue integrable function $z$ defined on $\Omega$:
$$
\overline{z}=\frac{1}{|\Omega|}\int_\Omega z(x) dx,
$$
where $|\Omega|$ denotes the Lebesgue measure of $\Omega$, using a finite random sample $X=\{x_{1},\dots, x_{n}\}$ of points in $\Omega$. Assuming that $\Omega$ is bounded, $|\Omega|$ is known and $X$ is a sampling process with inclusion density $\pi$, \cite{Cordy93} defines the continuous analogue of the Horvitz-Thompson estimator as:
$$
\widehat{\overline{z}}=\frac{1}{|\Omega|} \sum_{i=1}^{n} \frac{z(x_{i})}{\pi(x_{i})},
$$
and gives its properties. Under the assumption that $\pi(x)>0$ on $\Omega$ and that $z$ is bounded or non-negative, this estimator is unbiased \citep[][Theorem~1]{Cordy93}. If, moreover, $\int_{\Omega} 1/\pi(x) dx < +\infty$, the variance of this estimator is given by:
$$
\Var\left(\widehat{\overline{z}}\right) =\frac{1}{|\Omega|^{2}} \int_\Omega \frac{[z(x)]^{2}}{\pi(x)} dx + \int_\Omega \int_\Omega z(x)z(y) \left[ \frac{\pi^{(2)}(x,y)-\pi(x)\pi(y)}{\pi(x)\pi(y)}\right] dx dy,
$$

and if the joint inclusion density exists with $\pi^{(2)}(x,y)>0$ for all $x$, $y$ in $\Omega$ then:
\begin{equation}
\label{equ:varHH}
\widehat{\Var}\left(\widehat{\overline{z}}\right)=\frac{1}{|\Omega|^{2}} \sum_{x_{i}\in X} \left[\frac{z(x_{i})}{\pi(x_{i})}\right]^{2} + \sum_{x_{i}\in X} \sum_{\substack{x_{j}\in X\\ i \neq j}}z(x_{i})z(x_{j}) \left[ \frac{\pi^{(2)}(x_{i},x_{j})-\pi(x_{i})\pi(x_{j})}{\pi(x_{i})\pi(x_{j}) \pi^{(2)}(x_{i},x_{j})}\right],
\end{equation}
is an unbiased estimator of the variance of $\widehat{\overline{z}}$ \citep[][Theorem~2]{Cordy93}. As pointed out in \cite{Cordy93} the Horvitz-Thompson variance and variance estimator for a continuous population are slightly different from the finite population case. Conditions to ensure that these estimators are unbiased are, however, similar.

In the case of fixed size sampling process, the continuous analogue of the \cite{sen:53} and \cite{yat:gru:53} variance formula and estimator are:
\begin{equation}\label{SYG}
\Var\left(\widehat{\overline{z}}\right) = \frac{1}{2|\Omega|^{2}}\int_\Omega \int_\Omega \left[\frac{z(x)}{\pi(x)}- \frac{z(y)}{\pi(y)}\right]^{2} [\pi(x)\pi(y)-\pi^{(2)}(x,y)] dx dy,
\end{equation}
and 
\begin{equation}
\label{equ:varHHSYG}
\widehat{\Var}\left(\widehat{\overline{z}}\right)=\frac{1}{2|\Omega|^{2}} \sum_{x_{i}\in X} \sum_{\substack{x_{j}\in X\\ i \neq j}} \left[\frac{z(x_{i})}{\pi(x_{i})}-\frac{z(x_{j})}{\pi(x_{j})}\right] ^{2} \left[\frac{ \pi(x_{i})\pi(x_{j})-\pi^{(2)}(x_{i},x_{j})}{\pi^{(2)}(x_{i},x_{j})}\right],
\end{equation}
\citep[see][pp. 358-359]{Cordy93}. 

Throughout this paper, we assume that $\Omega\subset \mathbb{R}$ but the construction we used up to here also allows to work with other spaces. Indeed, \cite{Macchi75} and \cite{Cordy93} consider finite dimensional real vector spaces, and \cite{Daley02} work on complete separable metric spaces (polish spaces). Our purpose is to define sampling processes that have good properties regarding the estimation of $\overline{z}$.

In the following, we assume that $\Omega= (0,1)$. For an ordered set $\{x_{1},\dots,x_{n} \}$, we define the corresponding \emph{ inter-arrival times } $\{j_{1},\dots,j_{n-1} \}$ as the differences between two successive units, namely $j_{i} = x_{i+1}- x_{i}$, for $i = 1,\dots, n-1$. If $X$ is a point process, the corresponding inter-arrivals (also called waiting times) are random variables. A special class of point processes, called renewal processes, are obtained when the inter-arrival times are independent and identically distributed \citep[see for example][]{Mitov14}. In this paper, except when explicitly stated, the random inter-arrival times of our sampling processes are neither assumed to be identically distributed nor independent.

The binomial process \citep[see][pp. 23-28]{Moller03} is one of the most basic point processes and has a fixed sample size.
\begin{definition}[Binomial process]
Let $f$ be a PDF on $\Omega = (0,1)$ and let $n\in \mathbb{N}$ be a natural number. The \emph{binomial point process} of $n$ points in $\Omega$ with PDF $f$ is the point process whose realizations consist of $n$ points generated from i.i.d distributions with common PDF $f$.
\end{definition}
When the sample space $\Omega$ is bounded, inter-arrival times of the binomial process are not independent. Indeed, the sum of these inter-arrival times is necessarily no larger than the diameter of $\Omega$. In the following, we only use binomial processes in $(0,1)$ with i.i.d. points selected according to a uniform distribution on $(0,1)$.

The $k-$th order joint inclusion density of a binomial process of size $n$ at $x_{1}<\dots< x_{k}$ is given by:
$$
\pi^{(k)}(x_{1},\dots, x_{k}) = n(n-1)\cdots (n-k+1)= \frac{n!}{(n-k)!},\ k=1,\dots,n.
$$
In particular, $\pi(x_{i})=n$ if $x_{i} \in (0,1)$, and $\pi^{(2)}(x_{i},x_{j}) =n(n-1)$ if $x_{i},\ x_{j} \in (0,1)$. The $n-$th order joint inclusion density is equal to $n!$ on samples $x_{1},\dots,x_{n}$ with $0<x_{1}<x_{2}<\cdots<x_{n}<1$.

With $\Omega=(0,1)$ and a fixed size $n$, we can define the circular inter-arrival times as $J_{i}=(x_{i+1}-x_{i}) \mod 1$, $i=1,\dots,n-1$ and $J_{n}=(x_{1}-x_{n}) \mod 1$. As we see in Proposition~\ref{prop1}, the binomial process can be obtained by generating the circular inter-arrival times according a Dirichlet distribution. The Dirichlet distribution with parameter $\alphag$, denoted $\mathpzc{Dir}(\alphag)$ is a multivariate distribution with PDF given by
\begin{equation}
f(x_{1},\dots,x_{n}) = \frac{1}{B(\alphag)} \prod_{i=1}^{n} x_{i} ^{\alpha_{i}-1},
\label{dirichlet}
\end{equation}
where
$
x_{i}>0,
$
for $i,1,\dots,n$
$
\sum_{i=1}^{n} x_{i}=1,
$
$\alpha_{i}>0$, $\alphag=(\alpha_{1},\dots,\alpha_{n})$ and $\B(\alphag)$ is the multinomial Beta function. Properties of the Dirichlet distribution are given in \cite[][pp. 485-528]{kotzbalakrishnan2000continuous}. 
\begin{proposition}\label{prop1}
Let $\Jb^{c}=(J_{1}^{c},\dots,J_{n}^{c})\sim \mathpzc{Dir}({\bf 1}_{n})$, where ${\bf 1}_{n}$ is a vector of $n$ ones and $u\sim \mathpzc{U}(0,1)$, uniformely distributed on $(0,1)$, is independent from $\Jb^{c}$. The sorted values in
\begin{equation}\label{cccc}
\left\{\left( \sum_{j=1}^i J_{1}^c + u \right)\mod 1 ,\ i=1,\dots,n\right\},
\end{equation}
follow a binomial process on $(0,1)$ with uniform density.
\end{proposition}
\begin{proof}
With parameter ${\bf 1}_{n}$, the PDF in~(\ref{dirichlet}) simplifies to
$$
f_{\Jb^{c}}(j_{1},\dots,j_{n}) = (n-1)!,
$$
with $\sum_{i=1}^{n} j_{i}=1$. 
Let $\Xb=(X_{1},\dots,X_{n})$ be the sorted values~(\ref{cccc}). Since the sum of all $j_{i}$'s is equal to $1$, we see that a given set of numbers $x_{1}<\dots<x_{n}$ in $(0,1)$ is obtained exactly when $u={x_{i}}$ for some $i$ and the inter-arrival times allow to obtain $(x_{1},\dots,x_{n})$. These events are almost surely non overlapping and $u$ is independent from $\Jb^{c}$. It follows, if $f_{u}$ is the density of $u$ and $f_{\Jb^{c}}$ the density of $\Jb^{c}$, that
\begin{eqnarray*}
\lefteqn{ f_{\Xb}(x_{1},\dots,x_{n}) }\\
&=& f_{u}(x_{1})f_{\Jb^{c}}(x_{2}-x_{1},\dots ,x_{n}-x_{n-1}, x_{1}-x_{n}+1) \\
&+& f_{u}(x_{2})f_{\Jb^{c}}(x_{3}-x_{2},\dots, x_{1}-x_{n}+1, x_{2}-x_{1})\\
&\vdots&\\
&+& f_{u}(x_{n})f_{\Jb^{c}}(x_{1}-x_{n}+1,\dots, x_{n-1}-x_{n-2}+1, x_{n}-x_{n-1})\\
&=& n (n-1)!=n!.
\end{eqnarray*}
\end{proof}

The Poisson process \citep[see for example][]{Daley02, Moller03} is one of the basic and most studied point processes. It is particularly useful for the construction of more complex processes.
\begin{definition}[Poisson process]
A point process $X$ on $\Omega$ is a Poisson process with intensity $\lambda> 0 $ if the following properties are satisfied:
\begin{enumerate}
\item For any $A\in \mathcal{B}(\Omega)$, $N(X,A)$ follows a Poisson distribution with parameter $\lambda |A|$, where $|A|$ denotes the Lebesgue measure of $A$. If $|A|=0$, then $N(X,A)=0$ almost surely.
\item For any $n\in \mathbb{N}$, conditional on $N(X,A)=n$, the distribution of $X_{A}$ (the trace of the random set $X$ on $A$) is that of a binomial process on $A$ with size $n$ and constant PDF on $A$.
\end{enumerate}
\end{definition}
There exist several equivalent definitions of the Poisson process, but this one highlights the link with the binomial process. There is a similar link in finite population sampling, where conditioning a Bernoulli sampling design on its size yields a simple random sampling design \citep[see][pp. 43-50]{til:06}. Bernoulli sampling can thus be considered as the discrete analogue to the Poisson sampling process. Inter-arrival times of the Poisson process with intensity $\lambda$ are i.i.d. and follow an exponential distribution with parameter $\lambda$ \citep{Daley02}. 

It follows from the definition that the first order inclusion density of the Poisson sampling process on $\Omega = (0,1)$ is equal to $\lambda$, and using the independence property, that the $k-$th order joint inclusion density is equal to $\pi^{(k)}(x_{1},\dots, x_{k}) =\lambda^{k}$ if $x_{1} < \dots < x_{k}$.

The systematic process, or deterministic renewal process in the interval $(0,1)$ is defined as follows:
\begin{definition}[Systematic process]\label{def::sys_design}
Let $0<c<1$ and $u\sim \mathpzc{U}\left(0,c\right)$. A systematic sampling process with sampling interval $0<c<1$ is defined as the distribution of $\left\{x_{1},\dots,x_{n}\right\}$ where
$$
x_{k} = u+ k\cdot c,\quad k= 0,\dots, n-1,
$$
and $n$ is such that $u+n\cdot c < 1 \leq u+(n+1)\cdot c$.
\end{definition}

\section{Renewal processes}\label{S4}
A renewal process, or renewal sequence, is a stochastic process defined on the positive real line. It is completely characterized by the distribution of its independent and identically distributed inter-arrival times. For example, the Poisson process is a renewal process with exponentially distributed inter-arrival times when its intensity $\lambda$ is constant. The following definition can be found in \cite{Mitov14}.
\begin{definition}[Renewal process]
A renewal process is any process $X=\left\{ X_{k},\ k=0,\ 1,\ 2,\dots\right\}$ with 
$$
 X_{k} = X_{0} + \sum_{i=1}^{k} J_{i},\quad k=\ 1,\ 2,\dots
$$
where $X_{0}$ is a given non-negative random variable and $J_{1},\ J_{2},\dots$ is a sequence of i.i.d non-negative random variables with common Cumulative Distribution Function (CDF) $F$. If $X_{0}=0$ a.s., the process is called a pure renewal process (or simply a renewal process). If $P(X_0>0) >0$ then the process is called a delayed renewal process \citep[see][]{Resnick92}.
\end{definition}
The counting measure $N(t)$ (or renewal counting process) of a pure renewal process $X$ is defined in \cite{Mitov14} as:
$$
N(t)= \sup \left\{ k \geq 0 \mbox{ : } X_{k}\leq t\right\} = \sum_{i=1}^{\infty} \1_{\{X_{i} \leq t\}}, 
$$
where $\1_{\{X_{i} \leq t\} }$ denotes the indicator function. \cite{Daley02} then define the \emph{forward recurrence time} of a renewal process as:
$$
B(t) = X_{N(t)+1} -t, \quad t\geq 0.
$$
It is the random time between an arbitrarily chosen instant $t$ and the following occurrence of the process (see Figure~\ref{graph:lifetime_dist}).
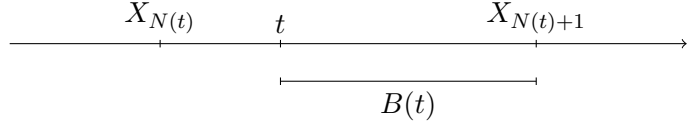
\begin{figure}[htb!]
\centering
\begin{tikzpicture}[scale=1]
\draw[->] (-4,0)--(5,0);
\draw plot (-2,0) node[above]{$X_{N(t)}$};
\draw (-2,0.05)--(-2,-0.05);
\draw plot (-.4,0) node[above]{$t$};
\draw (-.4,0.05)--(-.4,-0.05);
\draw plot (3,0) node[above]{$X_{N(t)+1}$};
\draw (3,0.05)--(3,-0.05);
\draw (-.4,-0.45)--(-.4,-0.55);
\draw (-.4,-0.5)--(3,-0.5) node[midway, below]{$B(t)$};
\draw (3,-0.45)--(3,-0.55);
\end{tikzpicture}
\caption{forward recurrence time $B(t)$}
\label{graph:lifetime_dist}
\end{figure}

An important result of renewal theory concerns the limiting distribution of the forward recurrence time $B(t)$ when $t\rightarrow\infty$. Under some mild conditions \citep[see][theorem 1.18]{Mitov14}, if the inter-arrival times have CDF $F$ and finite expectation $\mu >0$, $B(t)$ converges in distribution when $t\rightarrow\infty$ to a random variable with CDF $F_{0}$ defined as:
\begin{equation}
F_{0}(x)=\lim_{t\rightarrow\infty} P(B(t)\leq x) = \frac{1}{\mu} \int_{0}^{x} \left[1- F(t)\right] dt,\ x\geq 0.
\label{F0}
\end{equation}
The PDF of this limiting distribution is equal to:
$$ 
f_{0}(x)= \frac{1}{\mu} \left[1- F(x)\right],\ x\geq 0. 
$$
For example, if the inter-arrival times follow a Gamma distribution with shape parameter $r$ and rate parameter $\lambda$, denoted $\mathpzc{Gamma}(r,\lambda)$, their distribution function is given by $F(x)= \gamma(r, \lambda x)/\Gamma(r)$, where $\Gamma(r) = \int_{0}^{+\infty} t^{r-1}e^{-t} dt$ and $\gamma(r,x) = \int_{0}^{\lambda x} t^{r-1}e^{-t} dt$. The corresponding limiting forward recurrence time distribution follows a forward Gamma distribution $\mathpzc{ForG}(r,\lambda)$ with PDF:
$$
f_{0}(x) = \frac{\lambda \Gamma(r,\lambda x)}{\Gamma(r+1)},\ x\geq 0,
$$
with $\Gamma(r,\lambda x) = \int_{\lambda x}^{+\infty} t^{r-1}e^{-t} dt$. 

Another property of renewal processes that will be essential in the following is given in Proposition~\ref{fondamental}.
\begin{proposition}\label{fondamental}
Let $\left(J_{i}\right)_{i\geq 1}$ be a sequence of i.i.d non-negative continuous random variables, with expectation $\E(J_{i})=\mu$, CDF $F(x)$, and PDF $f(x)$. Let also $f_{0}$ be the function defined by:
$$
f_{0}(x) =\frac{1}{\mu} \left[1- F(x)\right] \mbox{ if }x\geq 0\mbox{ and } 0\mbox{ if }x<0.
$$
Then, equation~\ref{eqfond} holds
\begin{equation}\label{eqfond}
f_{0}(x) + \int_{0}^{x} f_{0}(x-t)\sum_{k=1}^{\infty} f^{k*}(t)dt= \frac{1}{\mu}, \mbox{ for all }x\geq 0,
\end{equation}
where $f^{k*}$ denotes the $k-$fold convolution of the function $f(x)$ with itself, i.e. the PDF of $\sum_{i=1}^{k}J_{i}$.
\end{proposition}
Proposition~\ref{fondamental} is a classical result of renewal process theory. We give a simple proof of it in appendix. Different proofs can be found for instance in \cite[p.47]{Mitov14} or in \cite[p.75]{Daley02}. Proposition~\ref{fondamental} implies that the delayed renewal process, obtained by generating $X_{0}$ with CDF $F_{0}$ and the $J_{i}$'s independently with CDF $F$, has, among other properties, a constant first-order inclusion density equal to $1/\mu$ on $\R_{+}$. Such a delayed renewal process has stationary increments and is called a stationary renewal process \citep{Mitov14}.

A special case is that of the Poisson process with intensity $\lambda$. It is a renewal process whose inter-arrival times follow an exponential distribution $\mathpzc{Exp}(\lambda)=\mathpzc{Gamma}(1,\lambda)$. It turns out that its limiting forward recurrence time distribution is also an exponential distribution with parameter $\lambda$, so that $F_{0}=F$. This is a consequence of the memory-less property of the exponential distribution.

\section{Quasi-systematic sampling}\label{S5}
Our aim is to propose new sampling processes that allow to control the selection probability of neighboring units by adjusting the joint inclusion density. Spreading the sample units over $\Omega$ has some advantages when units close together are similar (e.g. when the function $z$ has small variations). 

The systematic sampling process allows to select samples that are very well spread. However, it does not possess a positive second-order inclusion density so that \cite{Cordy93}'s Horvitz-Thompson variance estimator may not be used. We are thus interested in sampling processes with inter-arrival times that have a positive variance smaller than that of Poisson or binomial processes. Without auxiliary information that would encourage us to do otherwise, we focus on sampling processes with constant first-order inclusion density on $\Omega$.

The family of sampling processes that we consider can be seen as a compromise between basic sampling processes (Poisson and binomial processes) and systematic sampling. The rough idea is the following: in a first phase sampling procedure, a sample of expected size $n\cdot r$, with $n,r\geq 0$, is selected using an elementary sampling process. In the second selection phase, we use a systematic sampling to draw one unit every $r$ units of the first phase sample. We call these processes quasi-systematic sampling processes. We consider the ``systematic-Poisson'' and ``systematic-binomial'' processes obtained when the first phase processes are respectively the Poisson and the binomial process. The first and second-order inclusion densities of these sampling processes have a closed form.

Consider the following two-phases sampling process: a first phase sample is generated from a Poisson sampling process with constant intensity $\lambda$. Then, a systematic sample is drawn inside this first phase sample with rate $1/r$ (i.e. a starting unit is randomly chosen among the $r$ first units of the first phase sample and is kept in the second phase sample along with every other $r$ unit). In an interval of length $1$, the expected number of units selected by the Poisson process is $\lambda$. Thus, by setting $\lambda = n \cdot r $, where $n$ is the targeted final average sample size and $r$ is freely chosen, we ensure that the expected final sample size is $n$. 

The inter-arrival times of the first sample are, by definition, realizations of an exponential random variable with parameter $\lambda$. After the systematic sampling phase, inter-arrival times are realizations of sums of $r$ independent exponential random variables i.e. of non-negative random variables  $\mathpzc{Gamma}(r,\lambda)$ with PDF $f(x)=x^{r-1} e^{-\lambda x}\lambda^r/\Gamma(r)$. Thus, except for the first inter-arrival, this process is a renewal process with $\mathpzc{Gamma}(r,\lambda)$ renewal distribution. 

As we are set on having a constant first-order inclusion density, and thanks to Proposition~\ref{fondamental}, we choose to generate the first inter-arrival with a $\mathpzc{ForG}(r,\lambda)$ distribution and the following ones with independent $\mathpzc{Gamma}(r,\lambda)$ distributions. The first and second-order densities of the obtained systematic-Poisson sampling process are given in Proposition~\ref{denssp}. Note that parameters $n$ and $r$ do not in fact need to be integer numbers. Algorithm~\ref{alg:poiss_qs} can be used to select a systematic-Poisson sample in $(0,1)$.
\begin{algorithm}[htb!]
\caption{Generates a systematic-Poisson sample with parameters $\lambda$ and $r$.}
\label{alg:poiss_qs}
\begin{algorithmic}
    \REQUIRE $\lambda >0$, $r>0$;
    \STATE Generate $x_{1}\sim \mathpzc{ForG}(r,\lambda)$;
    \STATE i=2;
    \WHILE{$x_{i} < 1$}
		\STATE Generate $J_{i} \sim \mathpzc{Gamma}(r,\lambda)$
        \STATE $x_{i}= x_{i-1} + J_{i};i=i+1;$
        \IF{$x_{i}>1$}
        		\STATE $n=i-1$
        	\ENDIF
    \ENDWHILE
    \RETURN $\{x_{1},x_{2},\dots,x_{n}\}$ ordered systematic-Poisson sample with parameters $r$ and $\lambda$.
\end{algorithmic}
\end{algorithm}
\begin{proposition}\label{denssp}
Let us consider a systematic-Poisson process on $(0,1)$ with positive parameters $r$ and $\lambda$. Then
\begin{enumerate}
\item
the first-order inclusion density is given by:
$\pi(x) = \lambda/r, $
for any $x\in (0,1)$,
\item the second-order inclusion density is given by
\begin{equation}
 \pi^{(2)}(x,y) = \frac{\lambda}{r}
e^{-\lambda |x-y|} \sum_{m=1}^{\infty} \frac{\lambda^{mr}}{\Gamma(mr)}|x-y|^{mr-1},
\label{piklsystpoiss}
\end{equation}
for any $x,y \in (0,1)$.
\end{enumerate}
\end{proposition}
\begin{proof}\leavevmode
\begin{enumerate}
\item is a direct application of Proposition~\ref{fondamental}, considering that the expectation of a $\mathpzc{Gamma}(r,\lambda)$ distribution is equal to $r/\lambda$.
\item From \cite[p.139, Example 5.4(b)]{Daley02}, we have that, for $0\leq x<y$, 
$$
\pi^{(2)}(x,y)=\frac{\lambda}{r}u(y-x),
$$ 
where $u$ is the first-order density of the renewal process $X=(X_{i})_{i\geq 2}$, $X_{i} \sim \mathpzc{Gamma}(r,\lambda)$. $u(x)$ is equal to $\sum_{k=1}^\infty f^{k*}(x)$, where $x\geq 0$ and $f$ is the PDF of $X_{i}$. As the sum of $m$ independent $\mathpzc{Gamma}(r,\lambda)$ variables is a $\mathpzc{Gamma}(mr,\lambda)$ and has PDF:
$$
f(h;m)= \frac{\lambda^{mr}}{\Gamma(mr)}e^{-\lambda h } h^{mr-1},\mbox{ }h\geq 0,
$$
we can infer that the counting measure of the renewal process $X$ has renewal density  
$$
u(h) = \sum_{m=1}^{\infty} f(h;m) = \sum_{m=1}^{\infty}\frac{\lambda^{mr}}{\Gamma(mr)}e^{-\lambda h } h ^{mr-1},\mbox{ }h\geq 0,
$$
and the result follows.
\end{enumerate}
\end{proof}
The joint inclusion density equation simplifies for some values of $r$. Set $\lambda=n\cdot r$ with $n$ the expected the sample size. With $r=1$ we get the usual Poisson process and thus $\pi^{(2)}(x,y)=\lambda^{2}= n^{2}$.

The plot of $\pi^{(2)}(x,y)$ as a function of $|x-y|$ is given in Figure~\ref{G1} for different values of $r$. Except for $r=1$, $\pi^{(2)}(x,y)=0$ if $x=y$. The larger $r$ is, the flatter the plot is near the origin: the sampling design avoids selecting neighboring units. We see that, when $r$ is very large, the function concentrates on the inverse of the sampling rate and its multiples. It illustrates that the systematic-Poisson sampling design is close to a systematic sampling when $r$ is large.
\begin{figure}[htb!]
\begin{center}
\includegraphics[scale=0.6]{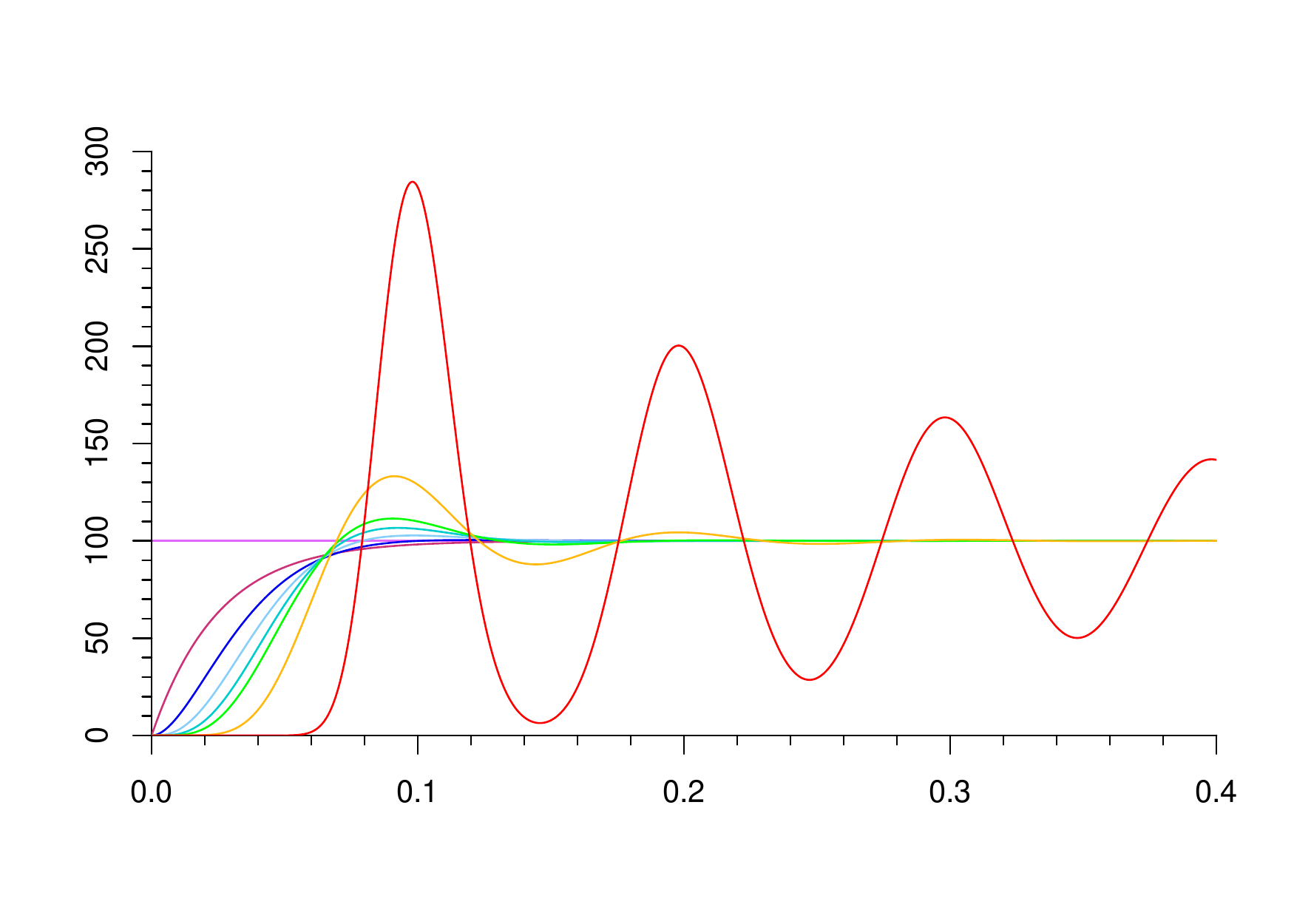}
\caption{\label{G1}Joint inclusion density $\pi^{(2)}(x,y)$ as a function of $|x-y|$ for systematic-Poisson sampling, for $n=10$, $r=1,2,3,4,5,6,10,50$ and $\lambda = n\cdot r$. The range of oscillations increases with $r$. When $r=1$, $\pi^{(2)}(x,y)$ is constant.}
\end{center}
\end{figure}

The systematic-binomial process is a fixed size sampling process with constant inclusion density on $(0,1)$. It is obtained, for example, by taking a realization of a binomial process of size $n\cdot r$, selecting a systematic sub-sample with rate $1/r$ inside the first phase units and finally adding, modulo~1, a random number $u$ generated from a $\mathpzc{U}(0,1)$ distribution. This last step ensures that the circular inter-arrival time $x_{1} + (1-x_{n})$ has the same distribution as the other inter-arrival times. An illustration of the sampling procedure is given in Figure~\ref{graph:sample_circle}.
\begin{figure}[htb!]
\centering
\begin{tikzpicture}[scale =2]
		\draw[color= gray, opacity = 0.6](0,0) circle (1);
		\draw[black, line width = .8pt] ( 95.583118731156:{1-0.03})--( 95.583118731156:{1+0.03} ) node[above]{{\footnotesize \textit{u}}};
		\draw [color= gray, opacity = 0.6]({-pi},-1.5)--({pi}, -1.5);
		\draw[black, line width = .5pt] ({-pi},-1.55)--({-pi}, -1.45) node[above]{{\footnotesize \textit{u}}};
		\draw[black, line width = .5pt] ({pi},-1.55)--({pi}, -1.45) node[above]{{\footnotesize \textit{u}}};
        \begin{axis}[
	    anchor=origin,  
	    x=1cm, y=1cm,   
	    hide axis
					]
      	\addplot[only marks,mark options={color=gray, opacity=.6}, mark size = 0.4pt] table  {pts_sample_comp_circle.csv};
        \addplot[only marks,mark options={color=gray, opacity=.6}, mark size = 0.4pt] table  {pts_sample_comp.csv};
        \addplot[only marks,mark options={color=red}, mark size = 0.4pt] table  {pts_sample.csv};
        \addplot[only marks,mark options={color=red}, mark size = 0.4pt] table  {pts_sample_circle.csv};

		\end{axis}
\end{tikzpicture}
\caption{\label{graph:sample_circle}Systematic-binomial sampling procedure with fixed size $n=10$ and $r=5$. In gray, the units sampled at the first phase, and in red the units in the final selection. On the top, we see the random shift $u$ plotted on a circle. On the bottom, we see the final sample on the interval $[0,1]$.}
\end{figure}
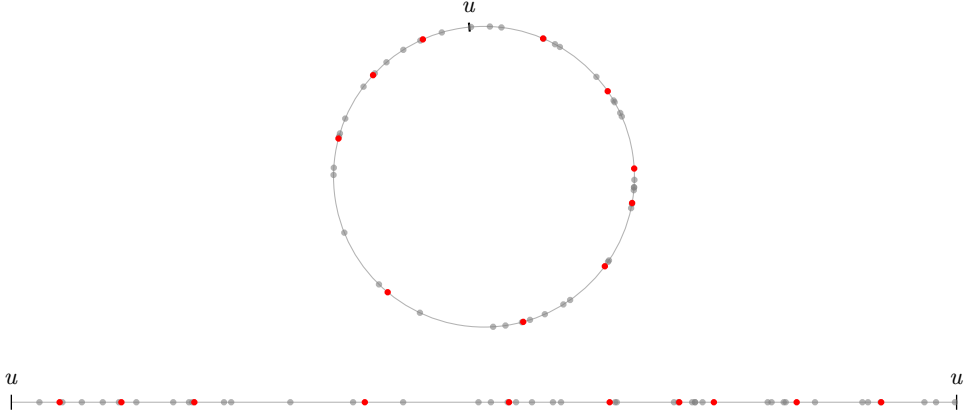
An implementation is proposed in Algorithm~\ref{alg:binom_qs}.
\begin{algorithm}[htb!]
\caption{Systematic-binomial sample with size $n$ and integer parameter $r$.}
\label{alg:binom_qs}
\begin{algorithmic}
   \REQUIRE  $n, r\in\mathbb{N}_{*}$.
\STATE Generate $\tilde{y}_{1},\dots,\tilde{y}_{nr}$ the sequence of order statistics of $n\cdot r$ i.i.d. variables $\mathpzc{U}(0,1)$.
    \FOR{$i=1,\dots,n,$} \STATE $\tilde{x}_{i} = \tilde{y}_{ir}$
    \ENDFOR
    \STATE Generate $u \sim \mathpzc{U}(0,1)$
    \FOR{ $i=1,\dots,n$} \STATE $x_{i}= (\tilde{x}_{i}+u) \mod 1$
    \ENDFOR
    \RETURN $\{x_{1},x_{2},\dots,x_{n}\}$ ordered systematic-binomial sample with parameter $r$ and size $n$.
\end{algorithmic}
\end{algorithm}

Another way to obtain a realization of a systematic-binomial process is to work with circular inter-arrival times. The first phase binomial sample is selected by generating $(\widetilde{J}^{c}_{i})_{i=1,\dots,nr}$, realization of a $\mathpzc{Dir}({\bf 1}_{nr})$ distribution, then these inter-arrival times are aggregated in packets of $r$ to form the circular inter-arrival times of the final sample,
\begin{equation}\label{dirn}
J_{i}^{c} = \sum_{k=1}^{r} \widetilde{J}^{c}_{(i-1)r+k},\ i=1,\dots, n,
\end{equation}
and finally a random uniform shift $u$ is used to set the origin. The selected units are
\begin{equation}
\label{rara}
\left(\sum_{j=1}^{i} J_{j}^{c} +u\right) \mod 1, \mbox{ for } i=1,\dots,\ n.
\end{equation}
However, the aggregation properties of the Dirichlet distribution ensure that the vector ${\bf J}^{c}=( J^{c}_{1},\dots, J^{c}_{n})$ of Equation~\ref{dirn} follows a $\mathpzc{Dir}(r{\bf 1}_{n})$ distribution. We also get that
\begin{equation}
\label{interarrisystbin1}
J^{c}_{i} \sim \mathpzc{Beta}( r,r(n-1) ),
\end{equation}
and
\begin{equation}
\label{interarrisystbin2}
\sum_{j=1}^{m} J^{c}_{i+j} \sim \mathpzc{Beta}(mr,mr(n-1)),\ 1\leq m\leq n-i-1,
\end{equation}
where $\mathpzc{Beta}(\cdot,\cdot)$ denotes the beta distribution. Taking advantage of this consideration, we can use Algorithm~\ref{alg:binom_qs2} to select samples from a systematic-binomial process. This method is not restricted to integer values of $r$. 
\begin{algorithm}[htb!]
\caption{Generate a systematic-binomial sample with size $n$ and real parameter $r>0$.}
\label{alg:binom_qs2}
\begin{algorithmic}
   \REQUIRE  $n \in\mathbb{N}_{*}, r \in\mathbb{R}_{+}^{*}$.
\STATE Generate ${\bf J}^{c} =( J^{c}_{1},\dots, J^{c}_{n} )\sim \mathpzc{Dir}(r{\bf 1}_{n}).$
    \STATE Generate $u \sim \mathpzc{U}(0,1)$
    \FOR{ $i=1,\dots,n$} \STATE $x_{i}= \left(\sum_{j=1}^{i} J_{j}^{c} +u\right) \mod 1$
    \ENDFOR
    \RETURN $\{x_{1},x_{2},\dots,x_{n}\}$ ordered systematic-binomial sample with parameter $r$ and size $n$.
\end{algorithmic}
\end{algorithm}

Inclusion densities of the systematic-binomial process are given in Proposition~\ref{densbin}.
\begin{proposition}\label{densbin}
Consider a systematic-binomial process of size $n$ with parameter $r$. Its inclusion densities are given below.
\begin{enumerate}
\item The first-order inclusion density is given by:
$$
\pi(x) = n,\mbox{ for } x\in (0,1). 
$$
\item The second-order inclusion density is given by
\begin{equation}
\pi^{(2)}(x,y)= n \sum_{m=1}^{n-1} \frac{\Gamma(nr)}{\Gamma(mr)\Gamma[(n-m)r]} |x-y|^{mr-1}(1-|x-y|)^{(n-m)r-1}, \label{pik2systbin}
\end{equation}
for $x \neq y \in (0,1)$.
\item The $n-$th order inclusion density is given by:
$$
\pi^{(n)}(x_{1},\dots,x_{n})=n\frac{\Gamma(nr)}{[\Gamma(r)]^{n}} (1+x_{1}-x_{n})^{r-1}(x_{2}-x_{1})^{r-1} \cdots \; (x_{n}-x_{n-1})^{r-1},
$$ 
for $x_{1}<\dots<x_{n} \in (0,1)$.
\end{enumerate}
\end{proposition}
\begin{proof}\leavevmode
\begin{enumerate}
\item Due to the random uniform shift used to set the origin, the point process canonically induced on the unit circle is clearly stationary (i.e. rotation invariant). Its first moment measure is thus a Haar measure and proportional to the Lebesgue measure. It follows that the first moment measure of the considered systematic-binomial process is proportional to the Lebesgue measure on $(0,1)$, and the proportionality coefficient is the total mass $n$.
\item The point process being stationary, its second-order inclusion density reduces to
$$
\pi^{(2)}(x,y)=n\cdot  u(y-x),\mbox{ if for example } 0\leq x<y<1,
$$
where $u$ is the first-order density of the point process $J^{c}_{2},\dots, J^{c}_{n}$ on $[0,1]$. However we have that the corresponding counting function $U(h)=N(0,h)$ is given by:
$$
U(h)=\sum_{m=1}^{n-1}F^{m*}(h),
$$
where $F^{m*}$ is the CDF of $\sum_{i=1}^{m} J^{c}_{2+i-1}$ and is thus the CDF of a $\mathpzc{Beta}(mr,mr(n-1))$ distribution. Hence 
$$
u(h)=\sum_{m=1}^{n-1} \frac{\Gamma(nr)}{\Gamma(mr)\Gamma[(n-m)r]} h^{mr-1}(1-h)^{(n-m)r-1},\mbox{ }0<h<1,
$$
and the result follows.
\item As with ordinary binomial sampling, a given sample is obtained exactly when $u$ is equal to one of the units and the inter-arrival times agree with the sample. Moreover the Dirichlet distribution with parameter $r{\bf 1}_{n}$ is symmetric and $u$ is independent from $\Jb^c$. We get that:
\begin{eqnarray*}
\lefteqn{ f_{\Xb}(x_{1},\dots,x_{n}) }\\
&=& f_u(x_{1})f_{\Jb^c}(x_2-x_{1},\dots ,x_{n}-x_{n-1}, x_{1}-x_{n}+1) \\
&+& f_u(x_2)f_{\Jb^c}(x_3-x_2,\dots, x_{1}-x_{n}+1, x_2-x_{1})\\
&\vdots&\\
&+& f_u(x_{n})f_{\Jb^c}(x_{1}-x_{n}+1,\dots, x_{n-1}-x_{n-2}+1, x_{n}-x_{n-1})\\
&=& n\frac{\Gamma(nr)}{\left[\Gamma(r)\right]^{n}} (1+x_{1}-x_{n})^{r-1}(x_2-x_{1})^{r-1} \cdots \;(x_{n}-x_{n-1})^{r-1} .
\end{eqnarray*}
\end{enumerate}
\end{proof}
Some straightforward computations lead to Equation~\ref{intpi2}
\begin{equation}\label{intpi2}
\int_{0}^{1}\int_{0}^{1} \pi^{(2)}(x,y) dx dy = n(n-1). 
\end{equation}
A plot of $\pi^{(2)}(x,y)$ as a function of $y$ is given in Figure~\ref{G2}, for $x=0.4$ and different values of $r$. Except for $r=1$, $\pi^{(2)}(x,y)=0$ if $x=y$. The larger $r$ is, the flatter the joint inclusion density is around $x=y$. The selection of neighboring units is thus very unlikely with such a sampling design and a large $r$. When $r$ is very large the function concentrates on regularly spaces pikes as in the systematic-Poisson case.
\begin{figure}[htb!]
\begin{center}
\includegraphics[scale=0.6]{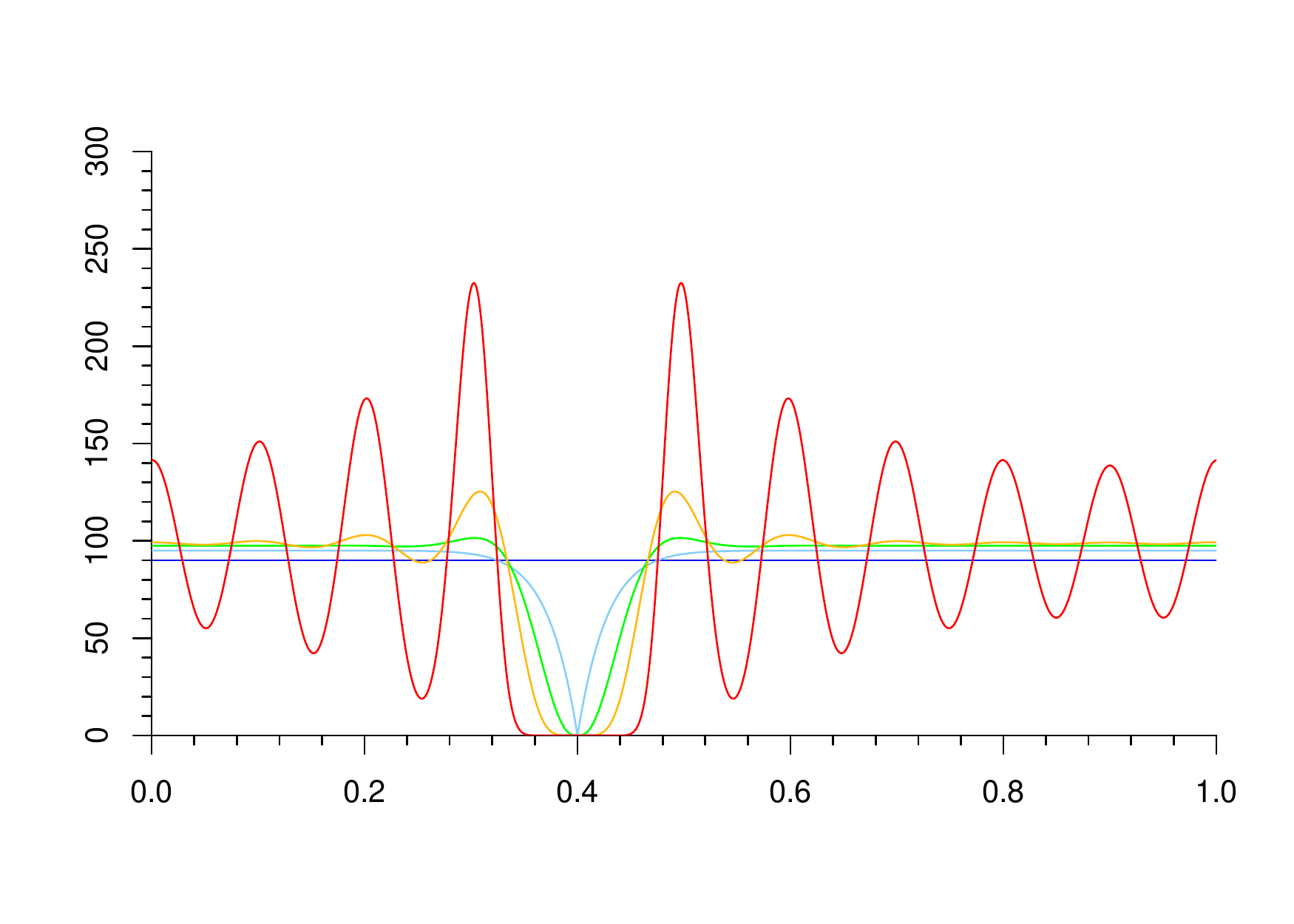}
\caption{\label{G2}Joint inclusion density $\pi^{(2)}(x,y)$ as a function of $y$ for $x=0.4$ in systematic-binomial sampling with $n=10$ and $r=1,2,4,8,30$. The range of oscillations increases with $r$.  When $r=1$, $\pi^{(2)}(x,y)$ is constant.}
\end{center}
\end{figure}

\section{Asymptotic results}\label{S6}
The sampling processes introduced in Section~\ref{S5} depend on a parameter $r$. When $r$ gets large, they look more and more like systematic sampling processes. Indeed, we will see that these processes converge in distribution to the systematic sampling process when $n$ is fixed and $r$ goes to infinity. We first need Lemma~\ref{convergenceb}.
\begin{lemma}
\label{convergenceb}
A forward gamma random variable $\mathpzc{ForG}(r,rn)$ converges in distribution to a continuous uniform variable $\mathpzc{U}(0,1/n)$ when $r$ tends to infinity and $n$ is fixed.
\end{lemma}
\begin{proof}
It is easy to prove that, if $\phi_{f}$ is the characteristic function of a positive probability distribution with expectation $\mu>0$, PDF $f$ and CDF $F$, then the characteristic function $\phi_{f_{0}}$ of the probability distribution with density $f_{0}=(1-F)/\mu$ is such that: 
$$
\phi_{f_{0}}(t)=\frac{1}{i\mu}\left[\frac{\phi_{f}(t)-1}{t}\right],\ t\in \R,
$$
where $i^{2}=-1$. However, the characteristic function of a $\mathpzc{Gamma}(r,\lambda)$ is given by $\phi_{\Gamma}(t)=(1-it/\lambda)^{-r}$. It follows that the characteristic function of a $\mathpzc{ForG}(r,\lambda)$ is given by
$$
\phi(t;r,\lambda)=\lambda \frac{ \left(\frac{\lambda}{\lambda - i t}\right)^{r}-1}{i r t},
$$
Replacing $\lambda$ by $rn$ and letting $r$ tend to infinity, we obtain that the characteristic function has a limit:
$$
\lim_{r\rightarrow\infty} \phi(t;r,rn)= \frac{e^{it/n}-1}{it/n},
$$
which is the characteristic function of a continuous uniform random variable $\mathpzc{U}(0,1/n) $. L\'{e}vy's continuity theorem applies and gives the result.
\end{proof}
We can now prove the announced result. We start with the systematic-Poisson process in Proposition~\ref{prop::convergence_syst_poiss}.
\begin{proposition}
\label{prop::convergence_syst_poiss}
Let us consider a systematic-Poisson process on $(0,1)$ with parameters $r >0 $ and $\lambda=rn$. Then, the process weakly converges to a systematic process of size $n$ when $r$ tends to infinity.
\end{proposition}
\begin{proof}
In systematic-Poisson process with parameter $r$ and $\lambda=rn$, the first inter-arrival time follows a forward Gamma distribution $\mathpzc{ForG}(r,rn)$ and the next ones follow a Gamma distribution $\mathpzc{Gamma}(r,rn)$. We have seen in Proposition~\ref{convergenceb} that $\mathpzc{ForG}(r,rn)$ converges to a $\mathpzc{U}(0,1/n)$ when $r$ tends to infinity. We also have that the $\mathpzc{Gamma}(r,rn)$ distribution converges to a $\mathpzc{Dirac}(1/n)$. Indeed, the expectation of a $\mathpzc{Gamma}(r,rn)$ is equal to $1/n$ and its variance to $1/(rn^{2})$. As the inter-arrival times are independent, we get that any finite family of them jointly converges to the matching distributions of inter-arrival times of a systematic process, as defined in Section~\ref{S3}. However, in the case of point processes the weak convergence of finite distributions is equivalent to the weak convergence of the process \citep[see, e.g., Theorem 11.1.VII of][p. 137]{Daley08}.
\end{proof}
The case of the systematic-binomial process is dealt with in Proposition~\ref{propbin}.
\begin{proposition}\label{propbin}
Consider a systematic-binomial process of size $n$ on $(0,1)$ and with parameter $r>0$. Then the process converges in distribution to a systematic sampling process when $r$ tends to infinity.
\end{proposition}
\begin{proof}
It is sufficient to show that the circular inter-arrival times converge in distribution to a $\mathpzc{Dirac}(1/n)$. Indeed, the random start is already accounted for in the procedure. However, the inter-arrival times follow a Beta distribution with mean $1/n$ and variance $r^{2}(n-1)/[(rn)^{2} (rn+1)]$, and indeed, their variance tends to $0$ when $r$ tends to infinity. As in the proof of  Proposition~\ref{prop::convergence_syst_poiss}, Theorem 11.1.VII in \cite{Daley08} allows to finish the proof.
\end{proof}

\section{Simulations}\label{S7}
Some simulations are useful to illustrate the properties of the systematic-binomial sampling process. We also ran simulations with the systematic-Poisson process and found that it behaves similarly but gives results that are less accurate than the systematic-binomial process with our test function. We considered the following test function:
$$
h(x) = 100 \sin\left( \frac{3x^{2}}{2x^{2} +1} \right) \exp\left\{-\left[\sin(4\pi x)^{2}\right]\right\},
$$
plotted in Figure~\ref{fig:est_binom} (left). We aim at estimating its mean using the Horvitz-Thompson estimator on a sample selected with a systematic-binomial process. A set of $10,000$ samples was generated using a systematic-binomial process with fixed size $n=30$ and  for each value of the parameter $r = 1,\ 2,\ 5,\ 10,\ 30,\ 50\mbox{ and } 100$.  Figure~\ref{fig:est_binom} (right) shows that the accuracy of the Horvitz-Thompson estimator increases with $r$. As expected, the systematic process performs better than any quasi-systematic process. 
\begin{figure}[htb!]
\centering
\includegraphics[height = 6cm]{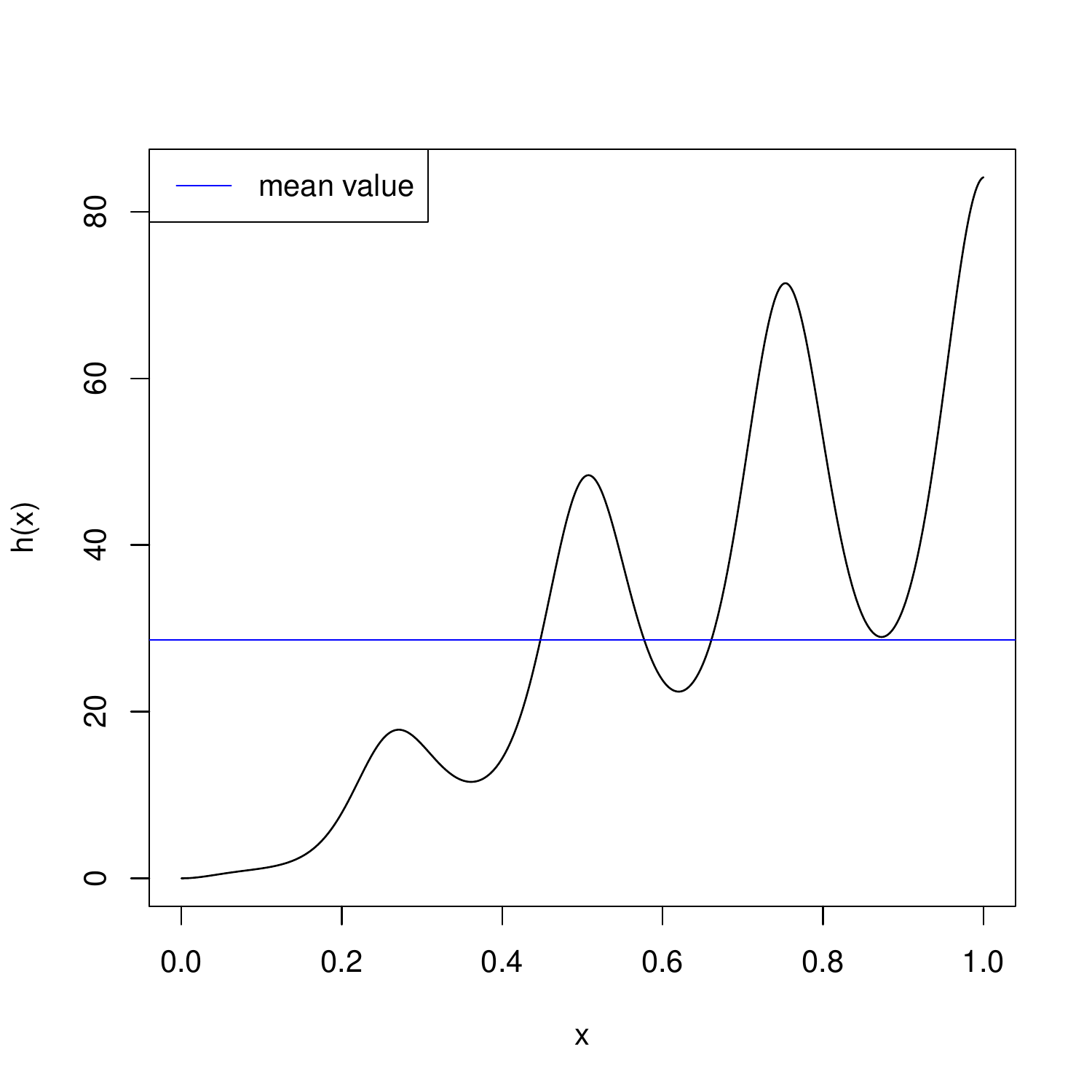}
\includegraphics[height = 6cm]{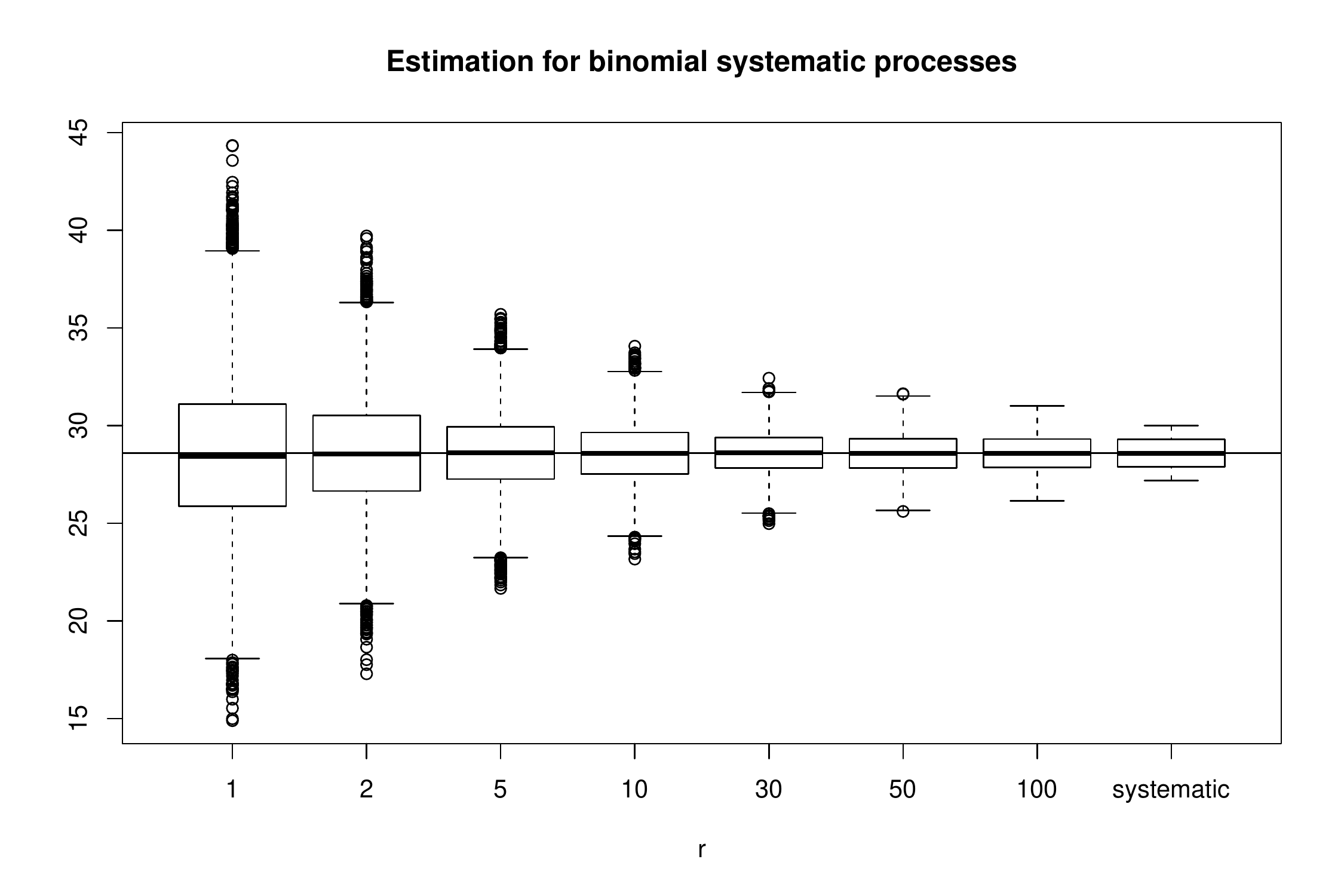}
\caption{\label{fig:est_binom}Test function (left) and boxplots of the estimated totals over all the  simulations (right). The parameter $r$ varies between $r=1$ and $r=100$ and we included the systematic sampling estimation. The horizontal line represent the true value of the total.}
\end{figure}
Corresponding simulation Root Mean Square Errors (RMSE) are given in Table~\ref{tab:RMSE}. We see in this table that the RMSE decreases rapidly with moderate values of $r$. 
\begin{table}[htb!]
\caption{RMSE of simulation results with a systematic-binomial process of size $n=30$ and different values of $r$.}\label{tab:RMSE}
\centering
\begin{tabular}{ccccccccc}
\hline
 & $r=1$ & $r=2$ & $r=4$ & $r=8$ & $r=30$ & $r=50$ & $r=100$ & Systematic \\
 \hline
 & 4.01 & 2.89 & 2.17 & 1.63 & 1.09 & 0.99 & 0.91 & 0.82 \\
 \hline
\end{tabular}
\end{table}

Estimating the variance of the Horvitz-Thompson estimator is a different issue. As previously stated, the variance estimator becomes unstable as $r$ increases, due to the fact that the second-order inclusion density tends to $0$ almost everywhere when $r$ goes to infinity. The estimated variance can also be negative in some cases. To alleviate these problems, the sample size $n$ should be increased when using large values of $r$. We give, in Table~\ref{tab:est.var} the mean over $10,000$ simulation samples of the variance estimator, their standard deviation as well as the true variance, for different combinations of $n$ and $r$. Since the systematic-binomial process has a fixed size, we use the Sen-Yates-Grundy variance estimator. The estimator is theoretically unbiased and decreases on average as the sample size $n$ increases when $r$ is fixed. We see that the standard deviation of the variance estimator values obtained in the simulations is consistently smaller for $r=2$ than for $r=1$ but gets a lot worse for larger values of $r$. Note that the simulation RMSEs of Table~\ref{tab:RMSE} mostly agree with the true variances in Table~\ref{tab:est.var}. 
\begin{table}[htb!] 
\centering
\caption{\label{tab:est.var} Estimated variance of a systematic-binomial process for different sample sizes $n$ and parameter $r$ values. In each cell, the simulation mean of the variance estimator with the corresponding Standard Deviations (SD) within parentheses, and the true target variance on the right.}%
\begin{tabular}{ccccccccc}
\hline
     & \multicolumn{2}{c}{$n=30$} & \multicolumn{2}{c}{$n=50$} & \multicolumn{2}{c}{$n=70$} & \multicolumn{2}{c}{$n=100$} \\
\hline
     & avg.$\widehat{\Var}$ (SD)  & $\Var$  & avg.$\widehat{\Var}$ (SD)  & $\Var$  & avg.$\widehat{\Var}$ (SD)  & $\Var$  & avg.$\widehat{\Var}$ (SD)  & $\Var$ \\
\hline
    $r=1$     & 15.90(3.60)   & 15.90  & 9.54(1.68)  & 9.53  & 6.82(1.00) &  6.81  & 4.76(0.59) & 4.76 \\
    $r=2$     &  8.53(2.41)   &  8.52  & 4.96(0.90)  & 4.97  & 3.51(0.57)  & 3.51  & 2.43(0.27) & 2.43 \\
    $r=4$     &  4.60(5.01)   &  4.66  & 2.62(2.15) &  2.62  & 1.82(0.89)  & 1.82  & 1.24(0.37) & 1.25 \\
    $r=8$     &  2.46(6.43)   &  2.68  & 1.43(6.25) &  1.44  & 0.95(1.32)  & 0.98 &  0.64(0.84) & 0.66 \\
    $r=30$    &  0.75(22.80)  &  1.21  & 0.30(7.72) &  0.56  & 0.15(3.45)  & 0.35 &  0.25(9.33) & 0.22 \\
\hline
\end{tabular}
\end{table}

We also see in table~\ref{tab:est.var} that the variance estimator gets very unstable for large values of $r$. One reason for this instability of the variance estimator is the joint inclusion density function getting close to 0 for large values of $r$ as can be seen on Figures~\ref{G1} and~\ref{G2}. Actually, this function, with $y$ in a neighborhood of a fixed $x$ in $(0,1)$, is driven by the first term in Equation~(\ref{pik2systbin}), and behaves like $|x-y|^{r-1}$. This is considered in Section~\ref{St} where we discuss the choice of the tuning parameter $r$. 

Another cause of instability in this example is that the test function has different values in $0$ and $1$ whereas the probability of jointly selecting $x>0$ but close to 0 and $y<1$ but close to 1 is small. When a sample is selected that contains such units, the variance estimator~(\ref{equ:varHHSYG}) takes a very large value. In our simulations, this case was responsible for most of the observed atypical very large values of the variance estimator. 

To solve this problem, if the test function $f$ is such that $f(0)\neq f(1)$, we define a new function $g$ by
$$
g(x)= \begin{cases} 
f(2x) &\mbox{if } 0\leq x \leq 1/2, \\ 
f(2x-2) & \mbox{if }   1/2< x\leq 1.
\end{cases} 
$$
The function $g$ is such that $\int g(x)\ dx = \int f(x)\ dx$ and satisfies $g(0) = g(1)$. As we see in Table \ref{tab:RMSE.sym}, replacing $f$ with $g$ does not increase the simulated RMSEs, nor the true variances found in Table~\ref{tab:est.var.sym}.
\begin{table}[htb!]
\caption{RMSE using the transformed function with a systematic-binomial process of size $n=30$ and different values of $r$.}\label{tab:RMSE.sym}
\centering
\begin{tabular}{ccccccccc}
\hline
 & $r=1$ & $r=2$ & $r=4$ & $r=8$ & $r=30$ & Systematic \\
 \hline
 &  4.00 &  2.94 &  2.09 & 1.47 & 0.76& 0.81 \\
 \hline
\end{tabular}
\end{table}

The variance estimator is much more stable with the transformed function $g$ than with the interest function $f$, as can be seen in Table~\ref{tab:est.var.sym}, compared with Table~\ref{tab:est.var}. The variance itself is slightly lower, meaning that the loss in spreading efficiency due to the transformation of the interest function is more than compensated by the absence of extreme values that were caused by $f(1)$ being different from $f(0)$.
\begin{table}[htb!] 
\centering
\caption{\label{tab:est.var.sym} Estimated variance of a systematic-binomial process using the transformed function for different sample sizes $n$ and parameter $r$ values. In each cell, the simulation mean of the variance estimator with the corresponding Standard Deviations (SD) within parentheses, and the true target variance on the right.}%
\begin{tabular}{ccccccccc}
\hline
     & \multicolumn{2}{c}{$n=30$} & \multicolumn{2}{c}{$n=50$} & \multicolumn{2}{c}{$n=70$} & \multicolumn{2}{c}{$n=100$} \\
\hline
     & avg.$\widehat{\Var}$ (SD)  & $\Var$  & avg.$\widehat{\Var}$ (SD)  & $\Var$  & avg.$\widehat{\Var}$ (SD)  & $\Var$  & avg.$\widehat{\Var}$ (SD)  & $\Var$ \\
\hline
    $r=2$     &  8.33(1.41)   &  8.31  & 4.86(0.60) &  4.85  & 3.43(0.36)  & 3.44 &  2.39(0.21) & 2.39 \\
    $r=4$     &  4.26(0.57)   &  4.26  & 2.44(0.23) &  2.44  & 1.72(0.13)  & 1.72 &  1.20(0.07) & 1.20 \\
    $r=8$     &  2.13(0.99)   &  2.15  & 1.23(0.25) &  1.22  & 0.86(0.09)  & 0.86 &  0.60(0.03) & 0.60 \\
    $r=30$    &  0.54(11.16)  &  0.58  & 0.28(2.69) &  0.33  & 0.20(1.19)  & 0.23 &  0.16(0.53) & 0.16 \\
\hline
\end{tabular}
\end{table}

When $r$ is not too large, confidence intervals exhibit coverage rates very close to the nominal rate of $95$\% as shown in Table~\ref{tab:cov_rate:sym}. These confidence intervals are computed assuming a normal approximation which seems compatible with our simulation results. However, for large values of $r$, $r\geq 30$ in our simulations, the estimation of the variance is very unstable, and the coverage rate of estimated confidence intervals deviates strongly. Indeed, for $r=30$ the low coverage rates in our simulations are explained by the variance estimator often taking negative values. In this case it would certainly be preferable to use a plain systematic process as the systematic-binomial process does not allow to get good confidence interval estimates. 
\begin{table}[htb!]
\centering
\caption{\label{tab:cov_rate:sym}Empirical coverage rates with a systematic-binomial sampling process and a transformed interest function, for different values of $n$ and $r$.}
\begin{tabular}{ccccc}
\hline
 & $n=30$ & $n=50$ & $n=70$ & $n=100$ \\
  \hline
$r=2$ & 0.9385 & 0.9473 & 0.9461 & 0.9479 \\
$r=4$ & 0.9422 & 0.9476 & 0.9469 & 0.9428 \\
$r=8$ & 0.9332 & 0.9474 & 0.9489 & 0.9513 \\
$r=30$& 0.4835 & 0.5102 & 0.5398 & 0.6019 \\
   \hline
\end{tabular}
\end{table}

\section{Choice of the tuning parameter}\label{St}

By choosing the tuning parameter $r$ one can make a compromise between an accurate estimation of the target parameter with a poor estimation of the precision and a less accurate estimation of the target parameter but with a reliable estimation of the estimator variance. Ideally one would have at its disposal a proxy interest function and could run simulations to select a suitable $r$, that is to say a $r$ that corresponds to one's preferred compromise. 

When no useful proxy function is available, some general remarks apply. Judging from our simulations, it seems that a small value of $r$ already helps reducing variance considerably compared to plain binomial process sampling. It is to be noted that, with values of $r$ between 1 and 2, the joint inclusion probability function $\pi^{(2)}(x,y)$ takes small values only when $x$ and $y$ are extremely close, as can be seen on Figures~\ref{G1} and~\ref{G2}. This is not the case anymore when $r$ is larger than $2$. In our simulations of Section~\ref{S7}, using the transformed function, we observed large values of the variance estimator only with $r$ larger than 2. 

A second point that could be inferred from our simulations is that larger sample sizes can accommodate for larger values of $r$. However, we do not have solid arguments to support that and we may just be lacking more simulation results here. It is to be noted though that, for fixed size processes such as the systematic-binomial process, one can check in advance which values of $r$ and sample size $n$, allow to satisfy the \cite{sen:53}, \cite{yat:gru:53} conditions: $\pi^{(2)}(x,y)\leq \pi(x)\pi(y)$ for all $x,y$. When these conditions hold, the variance estimator~(\ref{equ:varHHSYG}) is non-negative. Based on a numerical exploration, our conjecture is that this condition holds for $r=2$ and any sample size, but not for $r=3$. We also conjecture that, for fixed $r\geq 3$, increasing the sample size does not help reducing the maximal value of $\pi^{(2)}(x,y)/\pi(x)\pi(y)$. However, for a large enough $n$, and a given $x$, values of $y$ such that $\pi^{(2)}(x,y)/\pi(x)\pi(y)$ is greater than 1 are concentrated around $x$, and thus these couples do not contribute much to the variance estimator~(\ref{equ:varHHSYG}). Based on these considerations, it seems that $r=2$ could be a good compromise between stability of the variance estimator and stability of the target parameter estimator when no other information is available. The associated estimator true variance is however clearly greater than that obtained with larger values of $r$.

Finally, the regularity of the interest function has its importance. We can observe that having a function that satisfies a H\"{o}lder condition with exponent $\alpha\geq 0$ implies that the variance estimator~(\ref{equ:varHHSYG}) is bounded for all $r\leq 2\alpha+1$ (n.b.: we need to take the restriction of the function to $[0,1)$ and transport its source to the unit circle first in order to account for what happens near 0 and 1).

\section{Conclusion and discussion}\label{S8}
In this paper, we only worked on sampling processes with constant first-order inclusion density. It is however common in finite population survey sampling to choose different inclusion probabilities for different population units using auxiliary information available (e.g. the size of businesses or the approximate dispersion of the interest variable in a sub-population). Suppose we want to have a sampling process with first-order inclusion density proportional to a non-negative continuous function $\phi$, and note $\Phi(x)=\int_{0}^{x}\phi(t) dt$. Assume that the set of zeroes of $\phi$ have no interior, so that $\Phi$ is increasing. We just need to select a sample $x_{1},\dots,x_{n}$ with a constant inclusion density process, and retain $\Phi^{-1}(x_{1}),\dots,\Phi^{-1}(x_{n})$ as our sample. Indeed, if $\tilde{U}(x)=\E\{\tilde{N}([0,x])\}$ is the counting function of the new process and $U(x)=\E\{N([0,x])\}$ is the counting function of the constant density process, we have that 
$$
\tilde{U}(x)=U[\Phi(x)]=\lambda \Phi(x)\mbox{ for some }\lambda >0.
$$ 
It follows that $\tilde{U}(x)=\int_{0}^{x}\lambda \phi(t) dt$ and that the first-order inclusion density of the new process is given by $\tilde{\pi}(x)=\lambda \phi(x)$. The second inclusion density $\tilde{\pi}^{(2)}$ of this new process can also be derived from that, denoted by $\pi^{(2)}$, of the process used to select $x_{1},\dots,x_{n}$. We find that $\tilde{\pi}^{(2)}(x,y)=\pi^{(2)}[\Phi(x),\Phi(y)]\phi(x)\phi(y)$.

Both algorithms proposed in Section~\ref{S5} work with any positive value of $r$. The use of a parameter $0<r<1$ results in an attractive or clustering process where units tend to be selected in grouped clusters. This can be useful in some modelization problems. However, the interest of sampling with such clustering processes is probably limited to very specific objectives.

In future work, we intend to explore the possibility of developing similar sampling tools in spaces with more than one dimension. The generalization is far from being obvious as we only worked here on $\R$ equipped with its field ordering and some notions strongly depend on it.

Quasi-systematic sampling processes are useful to the practitioner who wants to make his own compromise between a more accurate estimation of a functions mean and a good estimation of the uncertainty of his estimator. Our simulations illustrate this trade-off between precision in the estimation of the mean and accuracy of the variance estimator. The former is better with a systematic sampling process while the latter is better with small values of $r$. We argue that quasi-systematic sampling processes could be used in place of plain binomial or Poisson processes for the purpose of estimating a mean in a continuous universe. A possible application is the estimation of the total or the mean of a variable of interest over time.

\section*{Acknowledgements}
The authors are grateful to one associate editor and three reviewers for their insightful comments that helped considerably improve the quality of this paper. This work was supported in part by the Swiss Federal Statistical Office. The views expressed in this paper are solely those of the
authors. M. W. was partially supported by a Doc.Mobility fellowship of the Swiss National Science Foundation.


\section*{References}
\begin{thebibliography}{20}
\expandafter\ifx\csname natexlab\endcsname\relax\def\natexlab#1{#1}\fi
\providecommand{\url}[1]{\texttt{#1}}
\providecommand{\href}[2]{#2}
\providecommand{\path}[1]{#1}
\providecommand{\DOIprefix}{doi:}
\providecommand{\ArXivprefix}{arXiv:}
\providecommand{\URLprefix}{URL: }
\providecommand{\Pubmedprefix}{pmid:}
\providecommand{\doi}[1]{\href{http://dx.doi.org/#1}{\path{#1}}}
\providecommand{\Pubmed}[1]{\href{pmid:#1}{\path{#1}}}
\providecommand{\bibinfo}[2]{#2}
\ifx\xfnm\relax \def\xfnm[#1]{\unskip,\space#1}\fi
\bibitem[{Baddeley \& Turner(2005)}]{spatstat}
\bibinfo{author}{Baddeley, A.~J.}, \& \bibinfo{author}{Turner, R.}
  (\bibinfo{year}{2005}).
\newblock \bibinfo{title}{\texttt{spatstat:} an {R} package for analyzing
  spatial point patterns}.
\newblock {\it \bibinfo{journal}{Journal of Statistical Software}\/},  {\it
  \bibinfo{volume}{12}\/}, \bibinfo{pages}{1--42}.
\bibitem[{Breidt(1995)}]{Breidt95}
\bibinfo{author}{Breidt, F.~J.} (\bibinfo{year}{1995}).
\newblock \bibinfo{title}{Markov chain designs for one-per-stratum sampling}.
\newblock {\it \bibinfo{journal}{Survey Methodology}\/},  {\it
  \bibinfo{volume}{21}\/}, \bibinfo{pages}{63--70}.
\bibitem[{Cochran(1977)}]{coc:77}
\bibinfo{author}{Cochran, W.~G.} (\bibinfo{year}{1977}).
\newblock {\it \bibinfo{title}{Sampling Techniques}\/}.
\newblock \bibinfo{address}{New York}: \bibinfo{publisher}{Wiley}.
\bibitem[{Cordy(1993)}]{Cordy93}
\bibinfo{author}{Cordy, C.~B.} (\bibinfo{year}{1993}).
\newblock \bibinfo{title}{An extension of the {H}orvitz-{T}hompson theorem to
  point sampling from a continuous universe}.
\newblock {\it \bibinfo{journal}{Statistics and Probability Letters}\/},  {\it
  \bibinfo{volume}{18}\/}, \bibinfo{pages}{353 -- 362}. \URLprefix
  \url{http://www.sciencedirect.com/science/article/pii/016771529390028H}.
  \DOIprefix\doi{http://dx.doi.org/10.1016/0167-7152(93)90028-H}.
\bibitem[{Daley \& Vere-Jones(2002)}]{Daley02}
\bibinfo{author}{Daley, D.}, \& \bibinfo{author}{Vere-Jones, D.}
  (\bibinfo{year}{2002}).
\newblock {\it \bibinfo{title}{An Introduction to the Theory of Point
  Processes: Volume I: Elementary Theory and Methods}\/}.
\newblock Probability and Its Applications (\bibinfo{edition}{2nd} ed.).
\newblock \bibinfo{address}{New York}: \bibinfo{publisher}{Springer}.
\bibitem[{Daley \& Vere-Jones(2008)}]{Daley08}
\bibinfo{author}{Daley, D.}, \& \bibinfo{author}{Vere-Jones, D.}
  (\bibinfo{year}{2008}).
\newblock {\it \bibinfo{title}{An Introduction to the Theory of Point
  Processes: Volume II: General Theory and Structure}\/}.
\newblock Probability and Its Applications (\bibinfo{edition}{2nd} ed.).
\newblock \bibinfo{address}{New York}: \bibinfo{publisher}{Springer}.
\bibitem[{Deville(1989)}]{dev:89x}
\bibinfo{author}{Deville, J.-C.} (\bibinfo{year}{1989}).
\newblock \bibinfo{title}{Une th\'eorie simplifi\'ee des sondages}.
\newblock In {\it \bibinfo{booktitle}{Les m\'enages : m\'elanges en l'honneur
  de Jacques Desabie}\/} (pp. \bibinfo{pages}{191--214}).
\newblock \bibinfo{address}{Paris}: \bibinfo{publisher}{INSEE}.
\bibitem[{Horvitz \& Thompson(1952)}]{hor:tho:52}
\bibinfo{author}{Horvitz, D.~G.}, \& \bibinfo{author}{Thompson, D.~J.}
  (\bibinfo{year}{1952}).
\newblock \bibinfo{title}{A generalization of sampling without replacement from
  a finite universe}.
\newblock {\it \bibinfo{journal}{Journal of the American Statistical
  Association}\/},  {\it \bibinfo{volume}{47}\/}, \bibinfo{pages}{663--685}.
\bibitem[{Kotz et~al.(2000)Kotz, Balakrishnan \&
  Johnson}]{kotzbalakrishnan2000continuous}
\bibinfo{author}{Kotz, S.}, \bibinfo{author}{Balakrishnan, N.}, \&
  \bibinfo{author}{Johnson, N.~L.} (\bibinfo{year}{2000}).
\newblock {\it \bibinfo{title}{Continuous Multivariate Distributions}\/}.
\newblock (\bibinfo{edition}{2nd} ed.).
\newblock \bibinfo{address}{New York}: \bibinfo{publisher}{Wiley}.
\bibitem[{Macchi(1975)}]{Macchi75}
\bibinfo{author}{Macchi, O.} (\bibinfo{year}{1975}).
\newblock \bibinfo{title}{The coincidence approach to stochastic point
  processes}.
\newblock {\it \bibinfo{journal}{Advances in Applied Probability}\/},  {\it
  \bibinfo{volume}{7}\/}, \bibinfo{pages}{83--122}.
\bibitem[{Madow \& Madow(1944)}]{mad:mad:44}
\bibinfo{author}{Madow, L.~H.}, \& \bibinfo{author}{Madow, W.~G.}
  (\bibinfo{year}{1944}).
\newblock \bibinfo{title}{On the theory of systematic sampling}.
\newblock {\it \bibinfo{journal}{Annals of Mathematical Statistics}\/},  {\it
  \bibinfo{volume}{15}\/}, \bibinfo{pages}{1--24}.
\bibitem[{Madow(1949)}]{mad:49}
\bibinfo{author}{Madow, W.~G.} (\bibinfo{year}{1949}).
\newblock \bibinfo{title}{On the theory of systematic sampling, {II}}.
\newblock {\it \bibinfo{journal}{Annals of Mathematical Statistics}\/},  {\it
  \bibinfo{volume}{20}\/}, \bibinfo{pages}{333--354}.
\bibitem[{Mitov \& Omey(2014)}]{Mitov14}
\bibinfo{author}{Mitov, K.~V.}, \& \bibinfo{author}{Omey, E.}
  (\bibinfo{year}{2014}).
\newblock {\it \bibinfo{title}{Renewal Processes}\/}.
\newblock Springer.
\newblock \bibinfo{address}{New York}: \bibinfo{publisher}{Springer}.
\bibitem[{M\o{}ller \& Waagepetersen(2003)}]{Moller03}
\bibinfo{author}{M\o{}ller, J.}, \& \bibinfo{author}{Waagepetersen, R.~P.}
  (\bibinfo{year}{2003}).
\newblock {\it \bibinfo{title}{Statistical Inference and Simulation for Spatial
  Point Processes}\/}.
\newblock \bibinfo{address}{London}: \bibinfo{publisher}{Chapman \& Hall/CRC}.
\bibitem[{M\o{}ller \& Waagepetersen(2007)}]{Moller07}
\bibinfo{author}{M\o{}ller, J.}, \& \bibinfo{author}{Waagepetersen, R.~P.}
  (\bibinfo{year}{2007}).
\newblock \bibinfo{title}{{Modern Statistics for Spatial Point Processes}}.
\newblock {\it \bibinfo{journal}{Scandinavian Journal of Statistics}\/},  {\it
  \bibinfo{volume}{34}\/}, \bibinfo{pages}{643--684}.
\bibitem[{Moyal(1962)}]{moy:62}
\bibinfo{author}{Moyal, J.} (\bibinfo{year}{1962}).
\newblock \bibinfo{title}{The general theory of stochastic population
  processes}.
\newblock {\it \bibinfo{journal}{Acta Mathematica}\/},  {\it
  \bibinfo{volume}{108}\/}, \bibinfo{pages}{1--31}.
\bibitem[{Resnick(1992)}]{Resnick92}
\bibinfo{author}{Resnick, S.~I.} (\bibinfo{year}{1992}).
\newblock {\it \bibinfo{title}{Adventure in Stochastic Processes}\/}.
\newblock \bibinfo{address}{Boston}: \bibinfo{publisher}{Birkh{\"a}user}.
\bibitem[{Sen(1953)}]{sen:53}
\bibinfo{author}{Sen, A.~R.} (\bibinfo{year}{1953}).
\newblock \bibinfo{title}{On the estimate of the variance in sampling with
  varying probabilities}.
\newblock {\it \bibinfo{journal}{Journal of the Indian Society of Agricultural
  Statistics}\/},  {\it \bibinfo{volume}{5}\/}, \bibinfo{pages}{119--127}.
\bibitem[{Till\'e(2006)}]{til:06}
\bibinfo{author}{Till\'e, Y.} (\bibinfo{year}{2006}).
\newblock {\it \bibinfo{title}{Sampling Algorithms}\/}.
\newblock \bibinfo{address}{New York}: \bibinfo{publisher}{Springer}.
\bibitem[{Yates \& Grundy(1953)}]{yat:gru:53}
\bibinfo{author}{Yates, F.}, \& \bibinfo{author}{Grundy, P.~M.}
  (\bibinfo{year}{1953}).
\newblock \bibinfo{title}{Selection without replacement from within strata with
  probability proportional to size}.
\newblock {\it \bibinfo{journal}{Journal of the Royal Statistical Society}\/},
  {\it \bibinfo{volume}{B15}\/}, \bibinfo{pages}{235--261}.

\end{thebibliography}
\end{document}